\documentclass[11pt, reqno]{amsart}
\usepackage{amsmath, amsthm, a4, latexsym, amssymb}

\def\@rmrk#1#2{\refstepcounter
    {#1}\@ifnextchar[{\@yrmrk{#1}{#2}}{\@xrmrk{#1}{#2}}}

\makeatletter\@addtoreset{equation}{section}\makeatother

\newfont{\bfit}{cmbxti10 scaled 2000}
\newfont{\biggi}{cmr12 scaled 2000}
\newtheorem{step}{STEP}

\newcommand{\bes}{\begin{step}}
\newcommand{\es}{\end{step}}
 \newcommand{\eps}{\varepsilon}

 \newcommand{\essinf}{{\rm essinf}\,}
 \newcommand{\R}{\mathbb{R}}

 \newcommand{\me}{\mathbb{E}}
 
 \renewcommand{\P}{\mathbb{P}}
 
 \newcommand{\skrib}{{\mathcal B}}
 \newcommand{\skric}{{\mathcal C}}

 \newcommand{\skrif}{{\mathcal F}}
 \newcommand{\skrig}{{\mathcal G}}
 \newcommand{\skrih}{{\mathcal H}}

 \newcommand{\skril}{{\mathcal L}}
 \newcommand{\skrim}{{\mathcal M}}

 \newcommand{\skrin}{{\mathcal N}}

 \newcommand{\skrix}{{\mathcal X}}

 \newcommand{\sfrac}[2]{\mbox{$\frac{#1}{#2}$}}

\def\1{{\mathchoice {1\mskip-4mu\mathrm l}      % Blackboard bold 1
{1\mskip-4mu\mathrm l}
{1\mskip-4.5mu\mathrm l} {1\mskip-5mu\mathrm l}}}

\newcommand{\eq}{\begin{equation}}
\newcommand{\en}{\end{equation}}

\renewcommand{\subsection}{\secdef \subsct\sbsect}
\newcommand{\subsct}[2][default]{\refstepcounter{subsection}
\vspace{0.15cm}
{\flushleft\bf \arabic{section}.\arabic{subsection}~\bf #1  }
\nopagebreak\nopagebreak}
\newcommand{\sbsect}[1]{\vspace{0.1cm}\noindent
{\bf #1}\vspace{0.1cm}}

\newtheorem{theorem}{Theorem}[section]
\newtheorem{lemma}[theorem]{Lemma}

\newtheoremstyle{thm}{1.5ex}{1.5ex}{\itshape\rmfamily}{}
{\bfseries\rmfamily}{}{2ex}{}

\newtheoremstyle{rem}{1.3ex}{1.3ex}{\rmfamily}{}
{\itshape\rmfamily}{}{1.5ex}{}
\theoremstyle{rem}

\refstepcounter{subsection}

\def\thebibliography#1{\section*{reference}
  \list%
  {\arabic{enumi}.}%                          *** style of reference number ***
    {\settowidth\labelwidth{[#1]}\leftmargin\labelwidth
    \advance\leftmargin\labelsep
    \parsep0pt\itemsep0pt
    \usecounter{enumi}}
    \def\newblock{\hskip .11em plus .33em minus .07em}
    \sloppy                   % \clubpenalty4000\widowpenalty4000
    \sfcode`\.=1000\relax}

\begin{document}
%%%%%%%%%%%%%%%%%%%%%%%%%%%%%%%%%%%%%%%%%%%%%%%%%%%%%%%%%%%%%%%%%%%%%%%%%%%%%%%
\title[Lossy  version  of  AEP for WSN]
{\Large Lossy  Asymptotic Equipartition Property for Geometric Networked Data Structures}

\author[K. Doku-Amponsah]{}

\maketitle
\thispagestyle{empty}
\vspace{-0.5cm}

\centerline{\sc{By Kwabena Doku-Amponsah}}
\renewcommand{\thefootnote}{}
\renewcommand{\thefootnote}{1}
\renewcommand{\thefootnote}{}
\footnote{\textit{Address:} Statistics Department, University of
Ghana, Box LG 115, Legon,Ghana.\,
\textit{E-mail:\,kdoku@ug.edu.gh}.}
\renewcommand{\thefootnote}{1}
%\centerline{\textit{Lecturer in Statistics}}%
\centerline{\textit{University of Ghana}}
%\centerline{\textit{E-mail:\,kdoku@ug.edu.gh}.}
%\centerline{\small Version: \jobname;  \version}\vspace{0.2cm}%

\begin{quote}{\small }{\bf Abstract.} This  article  extends the  Generalized  Asypmtotic  Equipartition  Property  of
Networked  Data Structures  to cover  the  Wireless Sensor Network   modelled  as  coloured geometric random  graph (CGRG).
The  main  techniques  used  to  prove  this result remains  large  deviation principles  for properly defined empirical
measures  on  CGRGs. As  a  motivation  for  this  article, we apply our  results to  some  data from  Wireless  Sensor Network for  Monitoring Water Quality from a Lake.

\end{quote}\vspace{0.5cm}

\section{Introduction}
Field data  we often  encounter from the  study  of the  environment  are  usually  structured  according to geometry and  the  connectivity  between  the  locations that  make  up  the  environment. Example,data  from  (i)  monitoring air quality
at key industrial sites, (ii) looking  for key contaminating
agents from the exhausts of public buses, (iii) monitoring the cleanliness in lakes  and  many more,  are all  structured according  to  the  geometry of the area  of  study and the  connectivity of the location that make up  the environment.   To design  and  implement  simplex (Linear  programming)  algorithm for the  solution of  generalized network flow problems of  the geometric structured network data ,see example \cite{ACS84}, or  to  find  an  efficient  coding  scheme or an  approximate  pattern  matching  algorithms,  see example \cite{CT91}, we need  an  information theory  for such data  structures, and  the  lossy  Asymptotic  Equipartition  Property  (AEP) for the geometric networked data structures is  key  to  finding  an information  theory  for the data  structure. See  \cite{DA16a} and  \cite{DA16b}  for  similar  results  for  other  types  of  data  structures.\\

 The  aim of this  article is  to  extend  the Lossy  AEP for Networked  Data  Structures modelled as Coloured  Random  Graph (CRG),
 see \cite[Theorem~2.1]{DA16a},  to cover the  WSN.
 To  be  specific  we  model  the   Geometric Networked Data Structures ( WSN)   as  a  CGRG and  use  some   of the
 large  deviation  techniques   developed  in  \cite{DA16b}  to  prove  a  strong  law  of  large  numbers (SLLN),
  see Lemma~\ref{WLLN7}, for  the  random network. Using  the   SLLN  and  the  techniques  deployed  in  \cite{DK02}
  we  extend the  Lossy  AEP  to  cover  the  WSN.\\

  The  remaining  part  of  the  paper  is  organized  as  follows: Section~\ref{Main} contains  the  main  result  of  the  paper  and  an  application to some data from
   environmental  science. See, Theorem~\ref{AEP1}  in Subsection~\ref{Main1}  and the  application in  Subsection~\ref{Main2}.
   Section~\ref{mainproof}  gives the  proof of  the  main  result; starting  with  the LDPs  (Lemmas~\ref{WLLN1}  and \ref{WLLN2})  in  Subsection~\ref{LDPs},   followed by  statement  and  proof  of  a  strong law  of  large  numbers,  see  Lemma~\ref{SLLN} and ending  with derivation  of  the  main  results  from  the  SLLN  in  subsection~\ref{AEPs}.

  \section{Generalized  AEP for CGRG  Process}\label{Main}

\subsection{Main Result}\label{Main1}

  We  consider  two  CGRG  processes   $X^{[z]}=\big\{(X(z_1), X(z_2)):\, z_i\,z_j\in E,\,i,j=1,2,3,...,n,i\not= j\big\}$  and
$Y^{[z]}=\big\{ (Y(z_i), Y(z_j)):\, z_i\,z_j\in E,\,i,j=1,2,3,...,n,i\not= j\big\}$  which  take  values  in  $G_{[z]}=G(\skrix,\, z_1,z_2,z_3,...,z_n)$  and
$\hat{G}_{[z]}=\hat{G}(\skrix,\, z_1,z_2,z_3,...,z_n),$ resp.,  the  spaces  of  finite graphs on
$\skrix$ and  $\, z_1,z_2,z_3,...,z_n\in [0,\,1]^d.$
We  equip  $G_{[z]}$, $\hat{G}_{[z]}$  with  their  Borel $\sigma-$fields $\skrif_{x}$
and  $\hat{\skrif}_{x}.$  Let  $\P_{x}$  and  $\P_{y}$  denote  the  probability measures  of the entire  processes $X^{[z]}$ and  $Y^{[z]}.$
 By $\P_{x}^{(\pi\omega)}$  and  $\P_{y}^{(\pi\omega)}$   we  denote  the   coloured  geometric random graphs $X^{[z]}$ and  $Y^{[z]}$ conditioned
 to have empirical colour measure  $\pi$ and empirical pair measure~$\omega.$ See, example  \cite{DA06}.
 We always  assume  that  $X^{[z]}$  and  $Y^{[z]}$  are  independent  of  each  other.\\

 By  $\skrix$  we  denote a  finite  alphabet  and  denote   by $\skrin(\skrix)$  the  space  of  counting  measure  on   $\skrix$
 equipped  with  the  discrete topology. By  $\skrim(\skrix)$  we  denote  the  space of  probability  measures  on  $\skrix$  equipped  with  the weak  topology  and  $\skrim_*(\skrix)$  the  space
 of finite  measures  on  $\skrix$  equipped  with  the  weak  topology.\\   %We always  assume  that  $X$  and  $Y$  are  independent  of  each  other.
%Recall  from \cite{DA16b}  that   $X=\big\{(X(u), X(v)):\, uv\in E\big\}$  and
%$Y=\big\{ (Y(u), Y(v)):\, uv\in E\big\}$   are  CRG processes  with  values from  $G=G(\skrix)$  and
%$\hat{G}=\hat{G}(\skrix),$ resp.,  the  spaces  of  finite graphs on  $\skrix.$

%We  define  the  process-level  empirical  measure  $\skril_n$  induced  by  $X$  and  $Y$  on  $G \times\hat{G}$   by
%$$\skril_n(\beta_x,\beta_y)=\frac{1}{n}\sum_{v\in [n]}\delta_{\big(\skrib_X(v),\,\skrib_Y(v)\big)}(\beta_x,\beta_y), \, \mbox{ for $(\beta_x,\beta_y)\in\skrim[{(\skrix\times\skrin(\skrix))}^2],$ }$$
%where  $\Psi(\beta_x(v),\beta_y(v))=\big((x(v),y(v)),\, c_{x,y}(v)\big).$
%Note  that  we  have
%$$\begin{aligned}
%\skril_n\otimes\Psi^{-1}\big((x(z_i),y(z_i)),\, l_{x,y}(z_i)\big)&=\frac{1}{n}\sum_{v\in [n]}\delta_{\big(\skrib_X(v),\,\skrib_Y(v)\big)}
%\big(\Psi^{-1}(x(z_i),y(z_i)),\, l_{x,y}(z_i)\big)\\
%&=\frac{1}{n}\sum_{v\in [n]}\delta_{\big((X(v),Y(v)),\, L_{X,Y}(v)\big)}\big((x(z_i),y(z_i)),\, l_{x,y}(z_i)\big)\\
%&:=\tilde{\skril}_n\big((x(z_i),y(z_i)),\, l_{x,y}(z_i)\big),
%\end{aligned}$$

%where  $\Psi(\beta_x(z_i),\beta_y(z_i))=\big((x(z_i),y(z_i)),\, l_{x,y}(z_i)\big).$

We define  the  process-level  empirical  measure  $\skril_{n,[z]}$  induced  by  $X^{[z]}$  and  $Y^{[z]}$  on   $G_{[z]}\times \hat{G}_{[z]}$ by
$$\skril_{n,[z]}(\beta_x(z),\beta_y(z))=\frac{1}{n}\sum_{v\in [n]}\delta_{\big(\skrib_X(z_v),\,\skrib_Y(z_v)\big)}(\beta_x(z),\beta_y(z)), \, \mbox{ for $(\beta_x(z),\beta_y(z))\in\skrim[{(\skrix\times\skrin(\skrix))}^2].$ }$$
%where  $\Psi(\beta_x(v),\beta_y(v))=\big((x(v),y(v)),\, c_{x,y}(v)\big).$

$$\skril_{n,[z],1}(\beta_x(z)):=\frac{1}{n}\sum_{v\in [n]}\delta_{\big(\skrib_X(z_v)\big)}(\beta_x(z))\,\mbox{and}\,  \skril_{n,[z],2}(\beta_y(z)):=\frac{1}{n}\sum_{v\in [n]}\delta_{\big(\skrib_Y(z_v)\big)}(\beta_y(z)) \, $$
 for $(\beta_x(z),\beta_y(z))\in\skrim[{(\skrix\times\skrin(\skrix))}^2].$

Throughout the  rest  of  the  article  we  will  assume that  $X^{[z]}$  and  $Y^{[z]}$   are   CGRG processes, See \cite{Pe98}.
For  $n\ge 1$,  let  $P_x^{(n)}$  denote  the  marginal  distribution  of  $X^{[z]}$  on  $[n]=\{1,2,3,...,n\}$  taking with  respect
to  $\P_{x}^{(\pi\omega)}$ and  $Q_y^{(n)}$  denote  the  marginal distribution    $Y^{[z]}$   on  $[n]=\{1,2,3,...,n\}$  with  respect  to $\P_{y}^{(\pi\omega)}.$ \\

Let  $\sigma:\skrix\times\skrin(\skrix)\times\skrix\times\skrin(\skrix)\to[0,\infty)$  be
an arbitrary  non-negative  function and  define  a  sequence of  single-letter  distortion  measures
$\sigma^{(n)}:G_{[z]}\times\hat{G}_{[z]}\to[0,\infty),$  $n\ge 1$  by
$$\sigma^{(n)}(x,y)=\frac{1}{n}\sum_{i\in [n]}\sigma\Big(\skrib_x(z_i),\,\skrib_y(z_i)\Big),$$

where $\skrib_x(z_i)=(x(z_i), L_x(z_i))$  and  $\skrib_y(z_i)=(y(z_i), L_y(z_i)).$  Given  $\alpha\ge 0$   and  $x\in G_{[z]}$ ,
 we  denote the  distortion-ball  of  radius  $\alpha$  by
$$B(x,\alpha)=\Big\{y\in\hat{G}_{[z]}:\,\, \sigma^{(n)}(x,y)\le \alpha\Big\}.$$

We  shall  call  the  measure  $\mu\in  \skrim[{(\skrix\times\skrin(\skrix))}^2]$  consistent  if  $ \mu_1$,  $\mu_2$  are  both  consistent   marginals  of  $\mu.$  Refer to  \cite[Equation~2.1]{DA16b}  for  the  concept  of  consistent  measures.

For  $(\pi,\,\omega)\in \skrim(\skrix)\times\skrim(\skrix\times\skrix),$  we write
$$p_{\pi\omega}(a,l)=\pi(a)\prod_{b\in\skrix}\frac{e^{-\omega(a,b)/\pi(a)}[\omega(a,b)/\pi(a)]^{\ell(b)}}{\ell(b)!},\,\mbox{for
$\ell\in\skrin(\skrix)$  } $$  and     define  the  rate  function  $I_1:\skrim[{(\skrix\times\skrin(\skrix))}^2]\to [0,\, \infty]$  by

\begin{equation}\label{AEP3}
\begin{aligned}
J_1(\mu)= \left\{ \begin{array}{ll}H\big(\mu\,\|\,p_{\pi\omega}\otimes p_{\pi\omega}), &
\mbox{if $\mu$  is  consistent  and $\mu_{1,1}=\mu_{1,2}=\pi$,}\\

 \infty & \mbox{otherwise,}

\end{array}\right.
\end{aligned}
\end{equation}

where  $$\,p_{\pi\omega}\otimes p_{\pi\omega}\big((a_x,a_y),(l_{x},l_{y})\big)
=p_{\pi\omega}(a_x,\,l_x)p_{\pi\omega}(a_y,\,l_y).$$ \\

  By  $x\,\approx\, p$  we  mean  $x$  has distribution  $p.$  For $(\pi,\,\omega)\in \skrim(\skrix)\times\skrim(\skrix\times\skrix),$  we  write
 $$\alpha_{av}(\pi,\omega)=\langle \log \langle e^{t\sigma(\skrib_X,\,\skrib_Y)},p_{\pi\omega}\rangle,p_{\pi\omega}\rangle.$$  %and
 Assume $$\alpha_{min}^{(n)}(\pi,\omega)=\me_{P_{x}^{(n)}}\big[\essinf_{Y\,\approx\,  Q_y^{(n)}}\sigma^{(n)}(X,Y)\big]\,\mbox{ $\to  \alpha_{min}(\pi,\omega).$}$$  For  $n>1,$   we  write  $$R_n(P_n^{(x)},Q_n^{(y)}, \alpha):=\inf_{V_n}\Big\{\frac{1}{n}H(V_n\,\|\,P_n^{(x)}\times Q_n^{(y)}):\,V_n\in \skrim(\skrig\times\hat{\skrig})\Big\}$$
and $$\alpha_{min}^{\infty}(\pi,\omega):=\inf\Big\{\alpha\ge 0:\,\sup_{n\ge 1}R_n(P_x^{(n)}, Q_y^{(n)}, \alpha)<\infty\Big\}.$$
Theorem~\ref{AEP1} (ii)  below  provides   a Lossy  AEP  for  WSN  data  structures.
\begin{theorem}\label{AEP1}
Suppose  $X^{[z]}$  and  $Y^{[z]}$   are  CGRG  process. Assume $\sigma$ are  bounded  function.  Then,
\begin{itemize}

\item[(i)] with  $\P_{(x)}^{(\pi\omega)}-$ probability $1,$ conditional  on  the  event $\big\{\,\Psi(\skril_{n,[z],1})=\Psi(\skril_{n,[z],2})=(\pi,\omega)\big\}$  the  random  variables  $\Big\{ \sigma^{(n)}(x,Y^{[z]})\Big\}$ satisfy  an  LDP   with  deterministic,  convex  rate-function  $$J_{\sigma}(t):=\inf_{\mu}\Big\{J_1(\mu):\, \langle\sigma, \,\mu\rangle=t\Big\}.$$
    %with  $\P_x-$ probability $1.$
\item [(ii)] for  all $\alpha\in\Big(\alpha_{min}(\pi,\omega),\,\alpha_{av}(\pi, \omega)\Big)$,  except  possibly  at $\alpha=\alpha_{min}^{\infty}(\pi, \omega)$
\begin{equation}\label{AEP11}
\lim_{n\to\infty}-\frac{1}{n}\log Q_x^{(n)}\Big(B(X^{[z]},\alpha)\Big)=R\big(\P_{x}^{(\pi\omega)},\P_{y}^{(\pi\omega)},\alpha\big)\,\,\mbox {almost  surely,}\end{equation}
where
$R(p,q,\alpha)=\inf_{\mu}H(\mu\,\|\,p\times q).$
\end{itemize}
\end{theorem}
\subsection{Application:}\label{Main2}{\bf Wireless  Sensor Network for  Monitoring Water Quality from a Lake.}  Let consider a WSN (to monitor the cleanliness in lakes,
particularly those used as sources of drinking water)  consisting  of   sensors  capable of carrying out some processing,
gathering sensory information and communicating with other connected nodes in the network  modelled  as  coloured geometric random  graph
  on  $n$ location, say $z_1,z_2,...,z_n.$  By $SG$   we  denote  sensors  capable of carrying out some processing,
gathering sensory information   while  communicating with  other  sensors and  $SI$  sensors  gathering  sensory  information while
communicating with  other  sensors. Suppose  the locations  are  $z_1,z_2,...,z_n\in [0,\,1]^d$  partition
into   $n\pi_n(SG)$  block of  $SG$  and  $n\pi_n(SI)$  block of
  $SI,$  and $n\|\omega_n^{\Delta(d)}\|$ number of  communication links divided into
   $n\omega_n^{\Delta(d)}(SG,\, SI),$  $n\omega_n^{\Delta(d)}(SI,\, SG),$ $ n\omega_n^{\Delta(d)}(SG,\, SG)/2,$ $ n\omega_n^{\Delta(d)}(SI,\,SI)/2$
   different interactions, respectively, for  $\Delta(d)$  a  function which  depends on the connectivity radius of the WSN. Assume $\pi_n$  converges $\pi$  and  $\omega_n^{\Delta(d)}$ converges  $\omega^{\Delta(d)}.$   If  we  take  $\sigma(s,r)=(s-r)^2$
        then, by  Theorem~\ref{AEP1}  we  have   the  rate-distortion  of

 \begin{equation}\label{AEP11}
\begin{aligned}
R(P,Q,\alpha)=\left\{ \begin{array}{ll} 0, & \mbox{ if \,$\alpha\ge 2\omega^{\Delta(d)}(SG,\, SI)+\omega^{\Delta(d)}(SG,\, SG)+\omega^{\Delta(d)}(SI,\,SI)
+2\omega^{\Delta(d)}(SI,\, SG)$.}\\

 \infty & \mbox{otherwise,}

\end{array}\right.
\end{aligned}
\end{equation}

 where $\omega^{\Delta(d)}(a,b)= \sfrac{\pi^{d/2}}{\big[d/2\big]!}\lambda_{[d]}(a,b))\pi(a) \pi(b).$   See, \cite{DA16a} for  the
relationship between  the connectivity radius and $\lambda_{[d]}.$ We  refer  to \cite{SB08} for more  on modelling of  the  physical  environment  using the  Wireless  Sensor  Network.

\section{Proof  of  Theorem~\ref{AEP1}.}\label{mainproof}

\subsection{LDPs.}\label{LDPs}\\

Recall  from \cite{DA16b}  that   $X=\big\{(X(u), X(v)):\, uv\in E\big\}$  and
$Y=\big\{ (Y(u), Y(v)):\, uv\in E\big\}$   are  CRG processes  with  values from  $G=G(\skrix)$  and
$\hat{G}=\hat{G}(\skrix),$ resp.,  the  spaces  of  finite graphs on  $\skrix.$

We  define  the  process-level  empirical  measure  $\skril_n$  induced  by  $X$  and  $Y$  on  $G \times\hat{G}$   by
$$\skril_n(\beta_x,\beta_y)=\frac{1}{n}\sum_{v\in [n]}\delta_{\big(\skrib_X(v),\,\skrib_Y(v)\big)}(\beta_x,\beta_y), \, \mbox{ for $(\beta_x,\beta_y)\in\skrim[{(\skrix\times\skrin(\skrix))}^2].$ }$$
%where  $\Psi(\beta_x(v),\beta_y(v))=\big((x(v),y(v)),\, c_{x,y}(v)\big).$
%Note  that  we  have
%$$\begin{aligned}
%\skril_n\otimes\Psi^{-1}\big((x(z_i),y(z_i)),\, l_{x,y}(z_i)\big)&=\frac{1}{n}\sum_{v\in [n]}\delta_{\big(\skrib_X(v),\,\skrib_Y(v)\big)}
%\big(\Psi^{-1}(x(z_i),y(z_i)),\, l_{x,y}(z_i)\big)\\
%&=\frac{1}{n}\sum_{v\in [n]}\delta_{\big((X(v),Y(v)),\, L_{X,Y}(v)\big)}\big((x(z_i),y(z_i)),\, l_{x,y}(z_i)\big)\\
%&:=\tilde{\skril}_n\big((x(z_i),y(z_i)),\, l_{x,y}(z_i)\big),
%\end{aligned}$$

%where  $\Psi(\beta_x(z_i),\beta_y(z_i))=\big((x(z_i),y(z_i)),\, l_{x,y}(z_i)\big).$

%  and  define  the  process-level  empirical  measure  $\skril_{n,[z]}$  induced  by  $X^{[z]}$  and  $Y^{[z]}$  on   $G_{[z]}\times \hat{G}_{[z]}$ by
%$$\skril_{n,[z]}(\beta_x(z),\beta_y(z))=\frac{1}{n}\sum_{v\in [n]}\delta_{\big(\skrib_X(z_v),\,\skrib_Y(z_v)\big)}(\beta_x(z),\beta_y(z)), \, \mbox{ for $(\beta_x(z),\beta_y(z))\in\skrim[{(\skrix\times\skrin(\skrix))}^2].$ }$$
%where  $\Psi(\beta_x(v),\beta_y(v))=\big((x(v),y(v)),\, c_{x,y}(v)\big).$

\begin{lemma}[Exponential Equivalence]\label{WLLN1}  Suppose   $(X^{[z]},Y^{[z]})$  are CGRG  on  the $d-$  dimensional  Torus
and  $(X,Y)$  are CRG.  Then,   conditional  on  the  event
$\big\{\,\Psi(\skril_{n,^[z],1})=\Psi(\skril_{n,[z],2})=\Psi(\skril_{n,1})=\Psi(\skril_{n,2})=(\pi, \omega)\big\}$   the  law  of
$\skril_{n,[z]}$ is  exponentially  equivalent  to
the  law  of    $\skril_n^.$
\end{lemma}

\begin{proof}

We  denote  by  $(\tilde{X},\tilde{Y})$   the  random  allocation process and   notice from  \cite[Lemma~3.1]{DA16a}    and  \cite[Lemma~0.4]{DA14}  that  conditional  on
$\big\{\,\Psi(\skril_{n,[z],1})=\Psi(\skril_{n,[z],2})=\Psi(\skril_{n,1})=\Psi(\skril_{n,2})=(\pi,\omega)\big\}$  the
law  of  $(X^{[z]},Y^{[z]})$  is exponentially  equivalent  to  the  law  of $(\tilde{X},\tilde{Y})$    and  the  law of  $(\tilde{X},\tilde{Y})$  is exponentially  equivalent  to  the  law  of  $(X,Y).$  Therefore, conditional  on  $\big\{\,\Psi(\skril_{n,[z],1})=\Psi(\skril_{n,[z],2})=\Psi(\skril_{n,1})=\Psi(\skril_{n,2})=(\pi,\omega)\big\}$  we  have  $(X^{[z]},Y^{[z]})$ exponentially  equivalent  to  $(X,Y).$

%$$d\Big(\skril_{n,[z]},\skril_{n}\Big)\le d\Big(\skril_{n,[z]},\tilde{\skril}_{n}\Big)
%+d\Big(\tilde{\skril}_{n}, \skril_{n}\Big),$$ the  triangle  inequality.

% As   $(X^{[z]},Y^{[z]})$  are  i.i.d  and  so  $(\skril_n^{+},\, \skril_n )$    are  i.i.d,  we  have
%$ \lim_{n\to\infty}\frac{1}{n}\log\prob\big\{d(\skril_{n,[z]} \,,
%\,\skril_n )\ge \sfrac{\eps}{2}\big\}\le\lim_{n\to\infty}\frac{1}{n}\log\prob\big\{\max\big(d(\skril_{n,[z]}
%\,,\,\tilde{\skril}_n ),d(\tilde{\skril}_n \,,\,{\skril}_n )\big)\ge\sfrac{\eps}{2}\big\}
%\le \lim_{n\to\infty}\frac{1}{n}\log\prob\big\{d({\skril}_{n,[z]} \,,\,\tilde{\skril}_{n} )
%\ge \sfrac{\eps}{2}\big\}+\lim_{n\to\infty}\frac{1}{n}\log\prob\big\{d(\tilde{\skril}_n \,,
%\,\skril_n )\ge \sfrac{\eps}{2}\big\}=-\infty-\infty=-\infty$  which  ends  the  prove  of  the  Lemma.

\end{proof}

\begin{lemma}[LDP]\label{WLLN2}  Suppose   $(X^{[z]},Y^{[z]})$  are coloured  geometric random  graph  on  the $d-$  dimensional  Torus.  Then,   conditional  on  the  event
$\big\{\,\Psi(\skril_{n,[z],1})=\Psi(\skril_{n,[z],2})=(\pi, \omega)\big\}$   the  law  of
$\skril_{n,[z]}$ obeys  a  process  level  LDP  with  good  rate  function  $J_1$
\end{lemma}

The proof  of  this  Lemma~\ref{WLLN2} which  follows  from \ref{WLLN1}  \cite[Theorem~]{DA16b}  and  \cite[Theorem~4.2.13]{DZ98}, is  omitted  from  the  paper.

\subsection{Derivation  of  the  AEP.}\label{AEPs} We  write  $\skrim:=\skrim[(\skrix \times \skrin(\skrix))^2]$  and  define  the  set  $\skric^{\eps}$  by
$$\displaystyle\skric_{\pi\omega}^{\eps}=\Big\{ \mu\in\skrim\colon
\sup_{\beta_x,\beta_y\in\skrix\times\skrin(\skrix)} |\mu(\beta_x,\,\beta_y) - p_{\pi\omega}\otimes p_{\pi\omega}(\beta_x,\,\beta_y)| \ge \eps\Big\}.$$

\begin{lemma}[SLLN]\label{WLLN7}\label{SLLN}  Suppose  the  sequence of  measures  $(\pi_n,\omega_n)$  converges to  the
pair  of  measures  $(\pi_n,\omega_n).$
For any $\eps>0$ we have  $\lim_{n\to\infty} \P_{(\pi_n,\omega_n)}\big(\skric_{\pi\omega}^{\eps}\big)=0.$

 \end{lemma}

Observe  that  $\skric_{\pi\omega}^{\eps}$  defined  above  is  a closed  subset  of  $\skrim$  and  so  by  Lemma~\ref{WLLN2}  we  have  that

\begin{equation}\label{AEP7}
\limsup_{n\to\infty}\frac{1}{n}\log\P_{(\pi_n,\omega_n)}\big(\skric_{\pi\omega}^{\eps}\big)\leq  -\inf_{\mu\in\skric^{\eps}}J_1(\mu).
\end{equation}

We use  proof  by contradiction  to  show that the right hand side of \eqref{AEP7} is negative.Suppose that there exists sequence  $\mu_n$ in $\skric_{\pi\omega}^{\eps}$ such that
$J_1(\mu_n)\downarrow 0.$ Then, there is a limit point $\mu\in F_1$ with $J_1(\mu)=0.$ Note $J_1$ is a good rate function and its level sets are compact,
and the mapping $\mu\mapsto J_1(\mu)$) lower semi-continuity. Now $J_1(\mu)=0$  implies $\mu(\beta_x,\,\beta_y)=p_{\pi\omega}\otimes p_{\pi\omega}(\beta_x,\,\beta_y),$  for  all   $\beta_x,\beta_y\in\skrix\times\skrin(\skrix) $ which contradicts  $\mu\in\skric_{\pi\omega}^{\eps}$.

(i)    Notice  $\displaystyle\sigma^{(n)}(X^{[z]},Y^{[z]})=\langle\sigma, \,\skril_{n,[z]}\rangle$  and  if  $\Lambda$  is open (closed)   subset  of  $\skrim$ then  $$ \Lambda_{\sigma}:=\big\{ \mu: \langle\sigma, \,\mu\rangle\in \Lambda\big\}$$ is  also open (closed) set  since  $\sigma$  is  bounded function.

$$\begin{aligned}
-\inf_{t\in In(\Lambda)}J_{\sigma}(t)&=-\inf_{\mu\in\ln(\Lambda_{\sigma})}J_1(\mu)\\
&\le\liminf_{n\to\infty}\sfrac{1}{n}\log\P \Big\{\sigma^{(n)}(X^{[z]},Y^{[z]})
\in\Lambda\big |X^{[z]}=x,\,\Psi(\skril_{n,[z],1})=\Psi(\skril_{n,[z],2})=(\pi_n,\omega_n)\Big\}\\
 &\le\lim_{n\to\infty}\sfrac{1}{n}\log\P \Big\{\sigma^{(n)}(X^{[z]},Y^{[z]})\in\Lambda\big |X^{[z]}=x,\,\Psi(\skril_{n,[z],1})=\Psi(\skril_{n,[z],2})=(\pi_n,\omega_n)\Big\}\\
 &\le\limsup_{n\to\infty}\sfrac{1}{n}\log\P\Big\{\sigma^{(n)}(X^{[z]},Y^{[z]})\in\Lambda\big |X^{[z]}=x,\,\Psi(\skril_{n,[z],1})
 =\Psi(\skril_{n,[z],2})=(\pi_n,\omega_n)\Big\}\\
 &\le -\inf_{\mu\in cl(\Lambda_{\sigma})}J_1(\mu)=-\inf_{t\in cl(\Lambda)}J_{\sigma}(t).
 \end{aligned}$$
% where  $I_{\sigma}(z):=\inf_{\omega}\Big\{J_1(\omega):\, \langle\sigma, \,\omega\rangle=z\Big\}.$

(ii)  Observe  that  $\sigma$ are  bounded, therefore  by  Varadhan's  Lemma and  convex duality, we  have
$$R( \P_x^{\pi\omega}, \P_y^{\pi\omega}, \alpha)=\sup_{t\in\R}[t\alpha-\skrih_{\infty}(t)]=\skrih_{\infty}^{*}(\alpha)$$
where
$$\skrih_{\infty}^{*}(t):=\lim_{n\to\infty}\sfrac{1}{n}\log \int e^{nt\Big\langle\sigma, \,\skril_{n,[z]}\Big\rangle}Q_y^{(n)}(dy)$$
exits  for  $\P$ almost  everywhere  $x.$  Using  bounded  convergence,  we  can  show  that
$$\skrih_{\infty}(t):=\lim_{n\to\infty}\skrih_n(t)=\lim_{n\to\infty}\sfrac{1}{n}\int \Big[\log \int e^{nt\Big\langle\sigma, \,
\skril_{n,[z]}\Big\rangle}Q_y^{(n)}(dy)\Big]P_x^{(n)}(dx).$$
Using   Lemma~\ref{WLLN7},  by  boundedness of  $\sigma$ we  have  that
$$\sfrac{1}{n}\skrih_n(nt)=\frac{1}{n}\sum_{j=1}^{n}\log\me_{Q_y^{(n)}}\big(e^{t\sigma(\skrib_x(j),\skrib_y(j)}\big)\to\langle
\log \langle e^{t\sigma(\skrib_{X^{[z]}},\skrib_{Y^{[z]}})},p_{\pi\omega}\rangle,p_{\pi\omega}\rangle=\alpha_{av}(\pi, \omega).$$
Also let  $$\alpha_{min}^{(n)}(\pi, \omega):=\lim_{t\downarrow-\infty}\sfrac{\skrih_n(t)}{t}$$
so  that  $\skrih_n^{*}(\alpha)=\infty$ for  $\alpha< \alpha_{min}^{(n)}(\pi, \omega)$,  while   $\skrih_n^{*}(\alpha)<\infty$ for  $\alpha>\alpha_{min}^{(n)}(\pi, \omega).$
 Observe  that  for  $n<\infty$  we  have   $\alpha_{min}^{(n)}(\pi, \omega)=\me_{P_x^{(n)}}\big[\essinf_{Y\,\approx\, Q_y^{(n)}}\sigma^{(n)}(X^{[z]},Y^{[z]})\big],$
 which  converges to  $\alpha_{min}(\pi, \omega).$
 Applying  similar  arguments  as  \cite[Proposition~2]{DK02}  we  obtain
  $$ R_n(P_x^{(n)},Q_y^{(n)},\alpha)=\sup_{t\in\R}\big(t\alpha-\skrih_n(t)\big):=\skrih_{n}^{*}(\alpha)$$

Now  we   observe  from \cite[Page 41]{DK02}  that  the  converge  of  $\skrih_{n}^{*}(\cdot)\to\skrih_{\infty}(\cdot)$  is
 uniform on  compact  subsets  of  $\R.$  Moreover, $\skrih_{n}$   is convex,  continuous  functions    converging  informally  to  $\skrih_{\infty}$  and  hence  we  can invoke  \cite[Theorem 5]{Sce48}  to  obtain

 $$\skrih_{n}^{*}(\alpha)=\lim_{\delta\to 0}\limsup_{n\to\infty}\inf_{|\hat{\alpha}-\alpha|<\delta}\skrih_{n}^{*}(\hat{\alpha}).$$

  Applying similar  arguments as \cite[Page 41]{DK02} in  the lines  after  equation  (64)  we  have  \eqref{AEP11}  which  completes  the  proof.

{\bf \Large Conflict  of  Interest}

The  author  declares  that  he has  no  conflict  of  interest.\\

{\bf \Large  Acknowledgement}

This  extension  has  been  mentioned  in the author's  PhD Thesis  at  University  of  Bath.

%%%%%%%%%%%%%%%%%%%%%%%%%%%%%%%%%%%%%%%%%%%%%%%%%%%%%%%%%%%%%%%%%%%%%%%%%%%%%%%

%{Kwabena Doku-Amponsah, University of Ghana, Department of
%Statistics, P.O. Box LG 115,\\ Legon-Accra, Ghana. E-mail:
%\texttt{kdoku@ug.edu.gh}.}


\begin{thebibliography}{WWW98}

\bibitem[1]{ACS84}
 {\sc I.~Ali}, {\sc A.~Charnes}  and {\sc T. Song}(1986).
 \newblock{Design  and  Implementation of Data  Structures for  Generalized  Networks.}
 \newblock{\emph{Journal of  Information  and  Optimization Sciences, 7(), 81-104.}}

\smallskip

\bibitem[2]{CT91}
 {\sc T.M.~Cover} and {\sc J.A.~Thomas}(1991).
 \newblock{Elements of Information Theory.}
 \newblock{Wiley Series in Telecommunications, (1991).}


\smallskip
\bibitem[3]{DA06}
{\sc K.~Doku-Amponsah.}(2006).
\newblock{Large deviations and  basic information theory for hierarchical and networked data structures.}
\newblock PhD Thesis, Bath (2006).
\smallskip

\bibitem[4]{DA10}
{\sc K.~Doku-Amponsah}(2012).
\newblock{Asymptotic equipartition propeties for hierarchical and networked  structures.}
\newblock{\emph{ ESAIM: PS 16 (2012): 114-138.DOI: 10.1051/ps/2010016.}}
\smallskip

\bibitem[5]{DA14}
{\sc K.~Doku-Amponsah}(2014).
\newblock{Exponential Approximation, Method of types for Empirical Neighbourhood Measures of Random graphs by Random Allocation.}
\newblock{\emph{ Int. J. Stat. \& Prob., 3(2),110-120 (2014).}}
\smallskip

\bibitem[6]{DA16a}
{\sc K.~Doku-Amponsah.}(2017).
\newblock{Lossy Asymptotic Equipartition Property  for Networked Data Structures.}
\newblock{\emph{J. Math. \& Stat. 13(2):152-158 •}}
\smallskip

\bibitem[7]{DA16b}
{\sc K.~Doku-Amponsah}(2015).
\newblock{Joint large deviation result for empirical measures of the coloured random geometric graphs.}
\newblock{ \emph{SpringerPlus.2016, 5:1140; Vol. 4, No. 1 (2015),pp. 87-93}}
\smallskip


\bibitem[8]{DA17b}
{\sc K.~Doku-Amponsah}(2017).
\newblock{Lossy Asymptotic Equipartition Property For Hierarchical Data Structures.}
\newblock{ \emph{Far East J. Math. Sc. 101(5):1013-1024.}}
\smallskip

\bibitem[9]{DK02}
{\sc A.~Dembo} and {\sc I.~Kontoyiannis}(2002).
\newblock Source Coding, Large deviations and Approximate Pattern.
\newblock \emph{Invited paper in IEEE Transaction on information Theory, 48(6):1590-1615,June (2002).}
\smallskip


\bibitem[10]{DZ98}
{\sc A.~Dembo} and {\sc O.~Zeitouni.}(1998)
\newblock Large deviations techniques and applications.
\newblock Springer, New York, (1998).
\smallskip

\bibitem[11]{Pe98}(1998).
{\sc D.B.~Penman.}
\newblock{Random graphs with correlation structure.}
\newblock PhD Thesis, Sheffield 1998.
\smallskip

\bibitem[12]{Sce48}
{\sc C.E.~Shannon}(1948).
\newblock{A Mathematical  Theory  of  Communication.}
\newblock {Bell  System Tech. J., 27:379-423,623-656.}
\smallskip

\bibitem[13]{SB08}
{\sc P. Stanley-Marbell},{\sc T. Basten},{\sc J. Rousselot},
{\sc R. S. Olive},{\sc H. Karl},{\sc M. Geilen},{R. Hoes},{\sc Gerhard Fohler} and {\sc J-D. Decotignie}(2008).
\newblock{System Models in Wireless Sensor Network.}
\newblock {\emph{ES Reports ISSN:1574-9517,  ESR-2008-06, 1 May 2008}}
\smallskip



\end{thebibliography}
\end{document}